\documentclass[11pt]{amsart}

\usepackage{times,amsfonts,amsmath,amstext,amsbsy,amssymb,
amsopn,amsthm,upref,eucal}
\usepackage[T1]{fontenc}
\usepackage{yfonts}
\usepackage{marginnote}

\newtheorem{theorem}{Theorem}[section]

\newtheorem {proposition}[theorem]{Proposition}
\newtheorem{corollary}[theorem]{Corollary}
\newtheorem{remark}[theorem]{Remark}

\newtheorem{example}[theorem]{Example}

\begin{document}
 
\markboth{Wojciech Czaja and James Tanis}{Kaczmarz Algorithm and Frames}

%%%%%%%%%%%%%%%%%%% Publisher's Area please ignore %%%%%%%%%%%%%%%%%%%%%%%
%\catchline{}{}{}{}{}
%%%%%%%%%%%%%%%%%%%%%%%%%%%%%%%%%%%%%%%%%%%%%%%%%%%%%%%%%%%%%%%%%%%%%%%%%%

\title{Kaczmarz Algorithm and Frames}

\author{Wojciech Czaja}

\address{Department of Mathematics, University of Maryland\\
College Park, MD 20742, USA\\
E-mail: wojtek@math.umd.edu}

\author{James H. Tanis}

\address{Department of Mathematics, Rice University\\
Houston, TX 77251, USA\\
E-mail: jtanis@rice.edu}

\maketitle

%\begin{history}
%\received{(Day Month Year)}
%\revised{(Day Month Year)}
%\accepted{(Day Month Year)}
%\comby{(xxxxxxxxxx)}
%\end{history}

\begin{abstract}
Sequences of unit vectors for which the Kaczmarz algorithm always converges in Hilbert space
 can be characterized in frame theory by tight frames with constant 1. We generalize this result to the context of frames and bases.  In particular, we show that the only effective sequences which are Riesz bases are orthonormal bases.   Moreover, we consider the infinite system of linear algebraic equations $A x = b$ and characterize the (bounded) matrices $A$ for which the Kaczmarz algorithm always converges to a solution.
\end{abstract}

%\keywords{Kaczmarz algorithm; Frame theory; Iterative algorithm.}

%\ccode{AMS Subject Classification: 22E46, 53C35, 57S20}

\section{Introduction}
In 1937 Stefan Kaczmarz \cite{K1}, \cite{K2}, introduced the following algorithm, known now as the {\it Kaczmarz algorithm} or {\it Kaczmarz method}, to solve a system of linear algebraic equations $A x = b$:

\noindent Let $a_{n}$ denote $n^{th}$ row of $A$, and let $x_{0} \in \mathbb{C}^{d}$. Define 
$$
x_{k + 1} = x_{k} + \frac{b_{i} - \langle x_{k}, a_i \rangle}{\|a_{i}\|^{2}}\,a_{i},
$$
where $k + 1 \equiv i \, (\text{mod} \,n)$.  If $A$ is of full rank, then:
$$
\lim_{k\to\infty} \|x_{k} - x\| = 0.
$$

This method has been rediscovered in the field of medical image processing, where it is used to reconstruct images in computed tomography and is called the Algebraic Reconstruction Technique (ART), see, e.g., \cite{GBH}, \cite{H}. Ever since, it has entered into many new research areas, such as crystallography, neural networks and parallel computing. Recently, the Kaczmarz algorithm has been combined with a randomization step, following on compressive sensing ideas, see \cite{CHJ}, \cite{SV1}, \cite{SV2}. We recommend \cite{C} for an updated list of publications involving the Kaczmarz algorithm.

McCormick \cite{McC} extended the Kaczmarz algorithm in 1977 to solve infinite systems of linear algebraic equations of the form $A x = b$, where $x, b$ are in Hilbert space.  His approach was based on an appropriate reduction of the infinite dimensional problem to a sequence of finite dimensional cases, where subiterations were performed on a sequence of increasing, finite dimensional subspaces.  

In 2001 Kwapie\'n and Mycielski \cite{KM} proposed a more straightforward version of the infinite dimensional Kaczmarz algorithm. Let $\mathbb{H}$ be a Hilbert space and let $\{e_{n}\}_{n=0}^{\infty}$ be a sequence of unit vectors in $\mathbb{H}$.  Given $x \in \mathbb{H}$, the Kaczmarz algorithm is defined as:
$$
x_{0} = \langle x, e_{0} \rangle \,e_{0},
$$
and
\begin{equation}\label{e20}
\forall\;  n\ge 1, \quad  x_{n} = x_{n - 1} + \langle x - x_{n - 1}, e_{n} \rangle e_{n}.
\end{equation}
The sequence $\{e_{n}\}_{n = 0}^{\infty}$ is called {\it effective} if and only if 
$$
\forall x \in\; \mathbb{H}, \quad \lim_{n\rightarrow \infty} x_{n} = x.
$$

Next, they introduced an algorithm that generates a sequence $\{g_n\}_{n = 0}^{\infty}$ from $\{e_n\}_{n = 0}^{\infty}$  with the property that $\{e_n\}$ is effective if and only if $\{g_n\}$ is a 1 - tight frame (see Theorem \ref{th1}).  In \cite{HS}, Haller - Szwarc later characterized effective sequences $\{e_n\}$, and therefore 1-tight frames $\{g_n\}$ containing a unit norm element, according to whether a certain matrix generated by $\{e_n\}$ is a partial isometry (see Theorem \ref{th2}).  

Until recently, the Kaczmarz algorithm has been considered to be a learning algorithm which allows for an infinite number of repetitive adjustments. This is clearly not optimal in many practical situations. As Strohmer and Vershynin have shown in their groundbreaking paper \cite{SV1}, a very different point of view can be taken, where a randomly selected subset of iterations is only considered. This reduces the repetitiveness of the algorithm, which, in turn, results in faster convergence. Usefulness of this approach has been further substantiated in \cite{CP}, \cite{N1}, \cite{N2}, \cite{N3}, \cite{N4}, \cite{N5}.

Motivated by these recent observations, we propose and analyze a different form of a constraint on the number of iterations in the Kaczmarz algorithm. This new constraint takes the form of Bessel property, which has been well studied in frame theory \cite{CR}. Associated with this we introduce a more flexible concept of almost effectiveness in order to study this correspondence in the context of frames, Riesz bases, and orthonormal bases. As an application we contribute to McCormick's work \cite{McC} on the infinite dimensional problem of solving $A x = b$ via the Kaczmarz algorithm.   

The paper is structured as follows: In Section \ref{s:2} we discuss the basics of frame theory and the algorithm introduced in \cite{KM} connecting the Kaczmarz algorithm with the theory of frames.  In Section \ref{s:3} we introduce and characterize almost effective sequences.  This enables us in Section \ref{s:4} to discuss the correspondence $\{e_n\}$ and $\{g_n\}$ in the context of frames, Riesz bases, and orthonormal bases.  In Section \ref{s:5} we consider these results in the context of the infinite dimensional problem of solving $A x = b$ via a single iteration of the Kaczmarz algorithm.

\section{Preliminary Results}
\label{s:2}

We start by introducing some basic terminology and notation that will be used throughout this paper.  These notions deal with the concept of redundant representations in Hilbert spaces. We say that a collection $\{ f_k : k \in \mathbb{N}\}\subset \mathbb{H}$ of vectors in a separable Hilbert space $\mathbb{H}$ is a {\it frame} for $\mathbb{H}$, if 
$$
\forall\; f \in \mathbb{H}, \quad
A\| f \|_2^2 \le \sum_{k\in \mathbb{N}} |\langle f,f_k \rangle|^2 \le B\| f \|_2^2,
$$
where $0 < A \leq B < \infty.$  Constants $A$ and $B$ which satisfy the above inequalities, are called, respectively, lower and upper frame bounds.  If $B < \infty$, then $\{f_k : k \in \mathbb{N}\}$ is a$\textit{ Bessel sequence}$.
We say that a frame $\{ f_k : k \in \mathbb{N}\}$ is {\it tight} if $A=B$, and a frame is called a {\it Riesz basis} if it is no longer a frame after removing of any of its elements. Riesz bases are also known as {\it exact frames}. Given any frame $\{ f_k : k\in \mathbb{N} \}$ for $\mathbb{H}$, there exists a {\it dual frame} $\{ {\tilde f}_k : k \in \mathbb{N} \}$ for $\mathbb{H}$ such that
\begin{equation}
\label{e1}
\forall\; f \in \mathbb{H}, \quad
f = \sum_{k\in \mathbb{N}} \langle f,f_k \rangle {\tilde f}_k
= \sum_{k \in \mathbb{N}} \langle f, {\tilde f}_k \rangle f_k,
\end{equation}
where the series converge in $\mathbb{H}$. The choice of coefficients for expressing $f$ in terms of  $\{ f_k : k \in \mathbb{N} \}$ or $\{ {\tilde f}_k : k \in \mathbb{N} \}$ is not unique, unless the frame is a {\it basis}. 
It is well known that {\it a frame is a basis if and only if it is exact}, see, e.g., \cite{BHW}.

Frames were introduced by Duffin and Schaeffer \cite{DS52} in 1952.  However, their practical potential has not been recognized until the 90's. We refer the interested reader to \cite{B92}, \cite{B94}, and \cite{CR}, for a more in depth treatment of frames and their constructions and applications. Since then, frames were, both, generalized and specialized, to allow for constructions of appropriately designed representation systems with varied features adapted to specific applications. 
This paper deals with one such special method of constructing frames. In 2001 Kwapie\'n and Mycielski \cite{KM} introduced the following sequence $\{g_{n}\}_{n=0}^{\infty} \subset \mathbb{H}$ to study effectiveness of the Kaczmarz algorithm:
$$
g_{0} = e_{0},
$$ 
and
\begin{equation}
\label{e2}
\forall \; n\in \mathbb{N}, \quad  g_{n} = e_{n} - \sum_{i = 0}^{n - 1} \langle e_{n}, e_{i}\rangle g_{i}.
\end{equation}
This specific construction allows us to write:
\begin{equation}
\label{e13}
x_n = \sum_{i=0}^n \langle x, g_i \rangle\, e_i.
\end{equation}
Kwapie\'n, Mycielski used this characterization to prove that 
\begin{equation}\label{eqn6}\|x\|^2 = \|x - x_n\|^2 + \sum_{n = 0}^{\infty} |\langle x, g_n\rangle|^2,
\end{equation}
This observation, in turn, leads to the following result.

\begin{theorem}[Kwapie\'n, Mycielski \cite{KM}]
\label{th1}
Let $\mathbb{H}$ be a separable Hilbert space and let $\{e_{n}\}_{n=0}^{\infty}$ be a sequence of unit vectors in $\mathbb{H}$. The sequence $\{e_{n}\}_{n = 0}^{\infty}$ is effective if and only if $\{g_{n}\}_{n = 0}^{\infty}$ constructed by means of (\ref{e2}) is a tight frame with constant 1 for $\mathbb{H}$.
\end{theorem}

We note here the following relationship between sequences $\{e_{n}\}_{n = 0}^{\infty}$ and $\{g_{n}\}_{n = 0}^{\infty}$:
\begin{equation}
\label{e2new}
e_n = \sum_{i=0}^{n} \langle e_n, e_i \rangle g_i.
\end{equation}
Equivalently, this relationship can be expressed as follows. Let M be the lower triangular matrix defined as:
$$
M(i,j) = m_{ij} = \begin{cases}\langle e_i, e_j \rangle & i>j,\\
1 & i=j,\\ 0 & i < j, \end{cases}
$$
and let $C = M^{-1}$ be the algebraic inverse of $M$ with coefficients defined by:
$$
C(i,j) = \begin{cases} c_{ij} & i>j,\\
1 & i=j,\\ 0 & i \le j. \end{cases}
$$
Then, we have that: 
\begin{equation}\label{e3}
g_n = e_n + \sum_{i=0}^{n-1}c_{n i}e_i.
\end{equation}

In \cite{HS} (see also \cite{H}), Haller and Szwarc obtained the following different characterization of effective sequences, which we will utilize in Section \ref{s:4}.

\begin{theorem}[Haller, Szwarc \cite{HS}]
\label{th2}
Let $\mathbb H$ and $\{e_n\}$ be as before. The sequence $\{e_{n}\}_{n = 0}^{\infty}$ is effective if and only if it is linearly dense and $C - I$ is a partial isometry, i.e., $(C-I)^*(C-I)$ is an orthogonal projection.
\end{theorem}

Proposition 1 in \cite{HS} proves the estimate $\|C - I\| \leq 1$, so that Theorem \ref{th2} can be interpreted as showing that effective sequences are as prevalent among sequences of unit vectors as partial isometries among strictly lower triangular contractions.%, and it provides a method for constructing effective sequences.

\begin{example} If $\{e_{n}\}_{n = 0}^{\infty}$ is an orthonormal basis for $\mathbb{H}$, then $g_n = e_n$. Hence, by Theorem 1, $\{e_{n}\}_{n = 0}^{\infty}$ is an effective sequence.
\end{example}

In view of Theorem \ref{th1}, and in view of (\ref{e1}) and (\ref{e2new}), it is natural to ask if $\{g_n\}$ is a dual frame of $\{e_n\}$. Such questions about duality arise naturally in the context of wavelet and Gabor frames \cite{E}, \cite{STB}.  
In Section \ref{s:4} we address affirmatively the question of whether the sequence $\{g_n\}$ is a frame.

The heart of our approach is the observation that the Bessel map of $\{g_n\}$ and that of $\{e_n\}$ are linked by the matrix $\overline{C} := (C^t)^*$.  Let $S_g = L_g^* L_g$ be the frame operator of $\{g_n\}$ with Bessel map $L_g$ and let $S_e$ be the frame operator of $\{e_n\}$ with Bessel map $L_e$.  Then we formally have
\begin{center}
$L_{g} x = \left(
\begin{array}{r} 
\langle x, g_{0} \rangle \\ 
\\
\langle x, g_{1} \rangle\\
\\
\langle x, g_{2} \rangle \\
\vdots 
\end{array}
\right) = 
\left(
\begin{array}{r} 
\langle x, e_{0} \rangle \\ 
\\
\sum_{i = 0}^{1} \overline{c}_{1 i} \langle x, e_{i} \rangle\\
\\
\sum_{i = 0}^{2} \overline{c}_{2 i} \langle x, e_{i} \rangle \\
\vdots 
\end{array}
\right)$= $\overline{C} L_{e} x$.
\end{center}

Theorems \ref{th8} and \ref{th9} describe a form of a frame-theoretic duality between the sequences $\{e_n\}$ and $\{g_n\}$.  However, the following remark shows these sequences are generally \emph{not} dual frames in the conventional sense.

\begin{remark}
The sequences $\{e_{n}\}_{n = 0}^{\infty}$ and $\{g_{n}\}_{n = 0}^{\infty}$ need not be dual frames of each other.
\end{remark}

\begin{proof}
Our argument is constructive. Let  $\{e_{n}\}_{n = 2}^{\infty}$ be the canonical orthonormal basis for $\overline{\textrm{span}}\{e_{0}, e_{1}\}^{\bot}$ and let $\langle e_{0}, e_{1}\rangle = \frac{1}{2}$.  Then, we conclude that $M = \left(\begin{array}{rrrr}
1 & & & \\
\frac{1}{2} & 1 & & \\
0 & 0 &1 & \\
0 & 0 & 0 & 1 \ddots \\
\end{array}\right)$, which, in turn, implies that 
$$
\overline{C} = C = \left(\begin{array}{rrrr}
1 & & & \\
-\frac{1}{2} & 1 & & \\
0 & 0 &1 & \\
0 & 0 & 0 & 1 \ddots\\
\end{array}\right).
$$

Note that $M$ is bounded and $\{e_{n}\}$ is a frame.  Under these conditions, Theorem \ref{th8} shows that $\{g_{n}\}$ is also a frame. We recall from the definition that if $\{e_{n}\}_{n = 0}^{\infty}$ and $\{g_{n}\}_{n = 0}^{\infty}$ were to form a dual pair, then the following would hold for every $f\in\ell^{2}(\mathbb{N})$:
$$
f = \sum_{n = 0}^{\infty}\langle f, g_{n}\rangle e_{n} = L_{e}\textrm{*}L_{g} f = L_{e}\textrm{*} \overline{C} L_{e} f.
$$  
However, $L_{e}\textrm{*} \overline{C} L_{e} e_{0} = L_{e}\textrm{*} \overline{C}(1, \frac{1}{2}, 0, \ldots)^{T} = L_{e}$*$(1, -\frac{1}{2}, 0, \ldots )^{T} = e_{0} - \frac{1}{2}e_{1} \neq e_{0}$.  Hence we have a contradiction.
\end{proof}

In particular, this example shows that even if we give $\{e_{n}\}_{n = 0}^{\infty}$ and / or $\{g_{n}\}_{n = 0}^{\infty}$ the extra structure of Riesz bases, they still may not be dual frames of each other.

\section{Almost Effective Sequences}
\label{s:3}

By considering frames $\{g_n\}_{n = 0}^{\infty}$ that are $\textit{not}$ tight, we are led through formula $\ref{eqn6}$ to define a sequence $\{e_{n}\}_{n = 0}^{\infty}$ to be \textit{almost effective} if there exists some $0 \leq B < 1$ such that 
$$\forall x\in H, \ \ \lim_{n \rightarrow \infty} \|x - x_{n}\|^{2} \leq B \|x\|^{2}.$$  

The condition of almost effectiveness is not only theoretically motivated by the concept of frames and frame inequalities, but it also arises naturally via the notions of numerical approximation and thresholding. In our opinion, a direct exploitation of this relaxation of convergence may lead to improvements in stability of implementations of the Kaczmarz algorithm.  In a similar spirit, weighted and controlled frames often improve the numerical stability of iterative algorithms that invert the frame operator \cite{ABG}.

We remark that the lower bound for almost effective sequences will always be zero, because if $x = e_{0}$, then $$\langle x, e_{0}\rangle e_{0} = e_{0}.$$  Now suppose that $x_{n} = e_{0}$.  Then $$x_{n + 1} = x_{n} + \langle x - x_{n}, e_{n + 1}\rangle e_{n + 1} = e_{0} + \langle e_{0} - e_{0}, e_{n + 1}\rangle e_{n + 1} = e_{0},$$ which means $$\lim_{n\rightarrow\infty} \|x - x_{n}\| = 0.$$  

\begin{theorem}
\label{th3}
Given is $0 < A \leq 1$. A sequence $\{e_{n}\}_{n = 0}^{\infty}$ is almost effective with bound 
$0 \leq (1 - A)$ if and only if $\{g_{n}\}_{n = 0}^{\infty}$ is a frame with bounds 
$0 < A \leq  1$.
\end{theorem}

We will now provide necessary and sufficient conditions for a Bessel sequence to be almost effective.  There are two necessary conditions, and the first is Proposition $\ref{l1}$ that is a consequence of Theorem $\ref{th3}$. 

\begin{proposition}
\label{l1}
Let $\{e_{n}\}_{n = 0}^{\infty}$ be a Bessel sequence that is also almost effective.  Then $\{e_{n}\}_{n = 0}^{\infty}$ is a frame.    
\end{proposition}

\begin{proof}  
Because $\{e_n\}$ is Bessel, we only need to show $S_e$ is bounded below.  By Proposition 1 in \cite{HS}, we know that $\|C - I\| \leq 1$, so $\|\overline{C}\| \leq 2.$
Say $\{e_{n}\}$ has an almost effective bound $(1 - A)$, then Theorem $\ref{th3}$ shows $\{g_{n}\}$ is a frame with lower bound $A > 0$.  Let $x \in \ell^{2}$.  Then
$$0 < A\|x\|^{2} \leq \langle S_{g} x, x \rangle = \langle \overline{C} L_{e} x, \overline{C} L_{e} x \rangle
\leq 4\langle L_{e} x, L_{e} x \rangle.$$

Therefore, 
\begin{equation}\label{e10}\langle S_{e} x, x \rangle = \langle L_{e}^{*} L_{e} x, x \rangle  \geq \frac{A}{4} \|x\|^{2} > 0,\end{equation}
which means $\{e_{n}\}$ is a frame.  
\end{proof}

In particular, every Bessel, effective sequence is a frame.  We now give the second necessary condition for a Bessel sequence to be almost effective.

\begin{theorem}
\label{th4}
Let $\{e_{n}\}_{n = 0}^{\infty}$ be a Bessel sequence that is also almost effective.  Then 
$\overline{C} = C^{t}\text{*}:\text{ran}(L_{e})\rightarrow \text{ran}(\overline{C}_{| \text{ran}L_{e}}$) is a Hilbert space isomorphism.
\end{theorem}

\begin{proof} 
By Proposition $\ref{l1}$, $\{e_n\}$ is a frame, which means ran($L_{e}$) is a Hilbert subspace of $\ell^{2}(\mathbb{N})$.

We show that $C^{t}\overline{C}$ is bounded below on ran($L_{e}$).  Let $y \in$Ran($L_{e}$), and let 
$x \in\ell^{2}(\mathbb{N})$ be such that $L_{e} x = y$.  Again because $\{e_{n}\}$ is a frame, there exists some constant $B > 0$ such that 
$\|y\| = \| L_{e} x\|^{2} \leq B \|x\|^{2}$.  Thus, $\|x\|^{2} \geq \frac{\|y\|}{B}$.  Moreover, almost effectiveness gives that $\{g_{n}\}$ is a frame.  So there exists some constant $A > 0$ such that 
$\|\overline{C} (L_{e} x) \|^{2} \geq  A \|x\|^{2}$.

Combining the two equations, it follows that 
$$ \| \overline{C} y \|^{2} \geq \frac{A}{B}\|y\|^{2},$$
which holds for all  $y \in$ ran$(L_{e})$.  So $C^{t} \overline{C} \geq \frac{A}{B} I > 0$ on this space.   

Now we show that $M$ is bounded on ran($\overline{C_{| \text{ran}L_{e}}}$) = $X$.  Because $M C = I$, we have $$\overline{MC} = \overline{I} = I.$$ 

Let $z \in X$, then take 
$y\in$ran($L_{e}$) such that $z = \overline{C} y$.  Then because 
$\| \overline{C} y \|^{2} \geq \frac{A}{B}\|y\|^{2},$ it follows that 
$$\|\overline{M} z\| = \|\overline{M} \overline{C} y\| = \| y\| \leq \frac{B}{A} \|\overline{C} y\|^{2} = \frac{B}{A} \|z\|^{2}.$$  
Therefore, $M$ is bounded on ran($\overline{C}_{| \text{ran}L_{e}}$), which implies 
$\overline{C}:\text{ran}(L_{e})\rightarrow \text{ran}(\overline{C}_{| \text{ran}L_{e}}$) is a isomorphism.
\end{proof}

Note that if a Bessel sequence $\{e_n\}$ is almost effective, then $M$ is bounded on ran($\overline{C}_{|\text{ran}{L_{e}}}$).  These two necessary conditions are also sufficient.

\begin{theorem}
\label{th6}
Suppose $\{e_{n}\}_{n = 0}^{\infty}$ is a frame and 
$\overline{C}: \text{ran}(L_{e}) \rightarrow \text{ran}({\overline{C}_{|\text{ran}L_{e}}})$ is a Hilbert space isomorphism.  Then $\{e_{n}\}_{n = 0}^{\infty}$ is almost effective.  
\end{theorem}

\begin{proof}
We show that $\{g_{n}\}$ is a frame.  Observe there exists constants $0< A_{1}, A_{2}, B_{1}, B_{2} < \infty$ such that 
$$\langle S_{g} x, x \rangle = \langle L_{e}^*C^{t}\overline{C}L_{e} x, x\rangle = \langle C^{t}\overline{C}L_{e} x, L_{e} x \rangle \geq A_{2}\| L_{e} x\|^{2} \geq A_{2} A_{1} \|x\|^{2},$$
and  
$$\langle C^{t}\overline{C}L_{e} x, L_{e} x \rangle \leq B_{2} \|L_{e} x\|^{2} \leq B_{2} B_{1} \|x\|^{2}.$$  So $\{g_{n}\}$ is a frame.  Then by Theorem \ref{th3}, $\{e_{n}\}_{n = 0}^{\infty}$ is almost 
effective.
\end{proof}

Proposition $\ref{l1}$ and Theorems $\ref{th4},$ $\ref{th6}$ now yield the following characterization for almost effective sequences.

\begin{theorem}
\label{th7}
Let $\{e_{n}\}_{n = 0}^{\infty}$ be a Bessel sequence.  Then $\{e_{n}\}_{n = 0}^{\infty}$ is almost effective if and only if it is a frame and $C: \text{ran}(L_{e})\rightarrow \text{ran}(C_{| \text{ran}L_{e}})$ is an isomorphism.
\end{theorem}

\section{Frames}
\label{s:4}

The results of this section can be viewed as an extension of Theorem \ref{th1} (proven in \cite{KM}) to the case when the sequence $\{e_n\}_{n = 0}^{\infty}$ is a frame, a Riesz basis or an orthonormal basis.  For the cases when $\{e_n\}_{n = 0}^{\infty}$ is a frame or a Riesz basis, our characterizations hold whenever $M$ is bounded.  This happens, for example, when there is some $N \in \mathbb{N}_0$ such that $\{e_n\}_{n = N}^{\infty}$ is orthonormal and more generally when $C - I$ is compact.

\begin{theorem}
\label{th8}
Let $\{e_{n}\}_{n = 0}^{\infty}$ be a Bessel sequence and let $M$ be bounded.  Then, $\{e_{n}\}_{n = 0}^{\infty}$ is a frame if and only if $\{g_{n}\}_{n = 0}^{\infty}$ is a frame.  
\end{theorem}

\begin{proof}
Because $M$ is bounded, injective, and $M^{-1} = C$ is bounded, Theorem \ref{th7} implies $\{e_n\}$ is a frame if and only if it is almost effective.  Then applying Theorem \ref{th3} proves Theorem \ref{th8}.
\end{proof}

We have a similar relationship between Riesz bases $\{e_{n}\}_{n = 0}^{\infty}$ and $\{g_{n}\}_{n = 0}^{\infty}$.

\begin{theorem}
\label{th9}
Let $\{e_{n}\}_{n = 0}^{\infty}$ be a Bessel sequence and $M$ be bounded.  Then $\{e_{n}\}_{n = 0}^{\infty}$ is a Riesz basis if and only if $\{g_{n}\}_{n = 0}^{\infty}$ is a Riesz basis. 
\end{theorem}

\begin{proof}

As in Theorem \ref{th8}, $M$ and $C$ are isomorphisms.

($\Leftarrow$)

Suppose $\{e_{n}\}$ is a Riesz basis.  Then it is a frame, so Theorem \ref{th8} proves $\{g_n\}$ is a frame.  Now note $L_{e}$ is surjective, so $L_{g} = \overline{C} L_{e}$ is also surjective, which means $ \{g_{n}\}$ is a Riesz basis. \\
($\Rightarrow$)

Conversely, suppose $\{g_{n}\}$ is a Riesz basis, so that $\{e_n\}$ is a frame, as before.  Because $L_g = \overline{C} L_e$ is surjective and $M$ is bijective, we conclude $$L_e = \overline{C}^{-1} \overline{C} L_e = \overline{M} L_g$$ is surjective.

\end{proof}

The next result, Theorem \ref{th10}, provides us with a simple characterization of all effective Riesz sequences in terms of orthonormal bases. This indicates that good sequences for the Kaczmarz algorithm must be found among inexact redundant frames.

\begin{theorem}
\label{th10}
An effective sequence is a Riesz basis if and only if it is an orthonormal basis.
\end{theorem}

\begin{proof}
Let $\{e_{n}\}$ be a non-orthonormal, effective sequence.  Then let $U = C - I$, and let $G$ be the Grammian of $\{e_{n}\}$.  Then $U$ and $G$ are bounded.  Moreover, it can be shown that $C$* $G \ C = I - U$* $U$ (Proposition 1, \cite{HS}).  Theorem $\ref{th2}$ shows    
$\{e_{n}\}$ is effective and if and only if $U$* $U$ is a projection.  
Then, because $\{e_{n}\}$ is not orthonormal, $M \neq I \Rightarrow C \neq I \Rightarrow U \neq$ 0.  So $U$ has at least one nonzero coordinate, say $c_{n i}$, and because $U$ is lower triangular, $n \geq i$.  Moreover, observe that the $(i i)$ position of $U$*$U$ is of the form are of the form 
$$\sum_{k = i}^{\infty} \mid c_{k i}\mid^{2}$$ which means $U$*$U$ has a nonzero entry in the $(i i)$ coordinate.  Hence, $U$*$U$ is a non-zero projection.

Then $I - U$*$U =  C$* $G \ C$ has a non-trivial kernel.  So there exists $x \neq$ 0 such that $0 = \langle C$* $G C x, x \rangle$ = $\langle G C x, C x \rangle$.  But $C$ is injective, so $C x \neq$ 0.  Hence, $G$ is not strictly positive,  
which means $\{e_{n}\}$ is not a Riesz basis.
\end{proof}

The following corollary (Corollary 4.4) is both known \cite{CZ} and elementary to prove. The brief proof we provide (one of several known to us) takes direct advantage of some of the observations made in this paper.  We present it here in order to illustrate a special case of and to provide context for the duality paradigm presented in Theorems 4.1 and 4.2.  As such, Corollary 4.4. is an example of the broader principle of duality that is at work here.  
\begin{corollary}
\label{c1}
Let $\{e_{n}\}_{n = 0}^{\infty}$ be a Bessel sequence.  Then $\{e_{n}\}_{n = 0}^{\infty}$ is an orthonormal basis if and only if $\{g_{n}\}_{n = 0}^{\infty}$ is an orthonormal basis.  
\end{corollary}

\begin{proof}

If $\{e_{n}\}$ is an orthonormal basis, then it follows immediately from definitions that $\{g_{n}\}$ is also.  Conversely, if $\{g_{n}\}$ is an orthonormal basis, then it is a 1-tight frame and a Riesz basis.  Thus, by Theorem \ref{th1} and Theorem \ref{th9}, $\{e_{n}\}$ is a Riesz basis that is also effective. Theorem \ref{th10} implies that it is an orthonormal basis. 
\end{proof}

\section{Applications}\label{s:5}

Recall that in \cite{McC} McCormick extended the classical Kaczmarz algorithm in order to find a solution $x \in \ell^2(\mathbb{N})$ to the infinite system of linear algebraic equations $A x = b$, where $A$ is a bounded operator on $\ell^2(\mathbb{N})$ and $b\in$ Ran$(A)$.  He reduced the infinite dimensional problem to a sequence of finite dimensional ones by subiterating the Kaczmarz algorithm on a sequence of increasing, finite dimensional subspaces.  This frame theoretic approach gives a convergence estimate without any subiterations, when we assume some frame conditions on rows of the operator $A$ in $A x = b$.

Let $A$ be an infinite dimensional matrix with linearly dense rows $\{a_{n}\}_{n = 0}^\infty\subset \ell^2(\mathbb{N})$ and $b, x \in \ell^2(\mathbb{N})$ be column vectors satisfying $A x = b$.  Notice that for all $n \geq 0$, we have $\langle a_n^*, x \rangle = b_{n}$, and define $e_n := \frac{a_n^*}{\|a_n\|}.$
Let \begin{equation}\label{e6}x_{0} = b_{0}\frac{a_{0}}{\| a_{0}\|^2}\end{equation} be the initial guess, so that $$x_0 = \langle a_0^*, x\rangle \frac{a_0^*}{\|a_0\|^2} = \langle e_{0}, x\rangle e_{0}.$$  
Now for $n \geq 0$, recursively define $$x_{n + 1} :=  x_{n} + \frac{b_{n + 1} - \langle a_{n + 1}^*, x_n \rangle}{\|a_{n + 1}\|^{2}} a_{n + 1}^* $$ \begin{equation}\label{e31}= x_{n} + \frac{\langle a_{n + 1}^*, x - x_n \rangle}{\|a_{n + 1}\|^{2}}a_{n + 1}^* = x_{n} + \langle e_{n + 1}, x - x_{n} \rangle e_{n + 1}.\end{equation}  It is clear the theorems in Section \ref{s:3} and \ref{s:4} also hold for this definition of the Kaczmarz algorithm.  

We derive the following two corollaries from our results in Sections $\ref{s:3}$ and $\ref{s:4}$.  Corollary $\ref{cor6}$ gives some quantitative information concerning how close this iterative scheme will get to satisfying $A x = b$, using the appropriate initial guess $x_0 = b_0 e_0.$  Corollary $\ref{cor7}$ characterizes the bounded operators $A$ for which the Kaczmarz algorithm always converges to a solution.  

\begin{corollary}\label{cor6}
Let $A : \ell^2(\mathbb{N})  \rightarrow \ell^2(\mathbb{N}) $ be a bounded (matrix) operator whose row vectors $\{e_n\}_{n = 0}^{\infty}$ have norm 1 and form a frame with upper and lower frame bounds $A_2 \geq A_1 > 0$, respectively, and let $b\in Ran(A)$.   Suppose $C - I$ is compact, so that there is some $C_1 > 0$ be such that $C \geq C_1 I$.  Then for the initial guess $x_0 = b_0 e_0,$ we have $$\lim_{n\rightarrow\infty}\|A x_{n} - b\|^{2} \leq A_2 \frac{(1 - A_{1} C_1)}{A_1}\|b\|^{2}.$$ 
\end{corollary}

\begin{proof}
Because $b\in Ran(A)$, there exists $x\in \mathbb{H}$ such that $A x = b$.  Notice $A = L_e,$ the Bessel map for $\{e_n\}$, so $0 < A_1 I \leq A^*A \leq A_2 I$ is the frame operator of $\{e_n\}$.  Then \begin{equation}\label{e8}\lim_{n\rightarrow\infty}\|A x_n - b\|^{2} =  \lim_{n\rightarrow\infty}\langle A^*A (x_n - x), (x_n - x)\rangle^{2} \leq A_2 \lim_{n\rightarrow\infty}\|x_n - x\|^{2}.\end{equation}  Now $C - I$ compact implies $M$ is bounded (by the Fredholm alternative).  Then because $\{e_n\}$ is a frame, Theorem \ref{th6} shows it generates a frame $\{g_n\}$ with a lower frame bound $A_1 C_1 > 0$.  Finally, Theorem $\ref{th3}$ shows $\{e_n\}$ is almost effective with bound $1 - A_1 C_1$, so \begin{equation}\label{e14}(\ref{e8}) \leq A_2 (1 - A_1 C_1) \|x\|^{2}.\end{equation}  

Finally, because $A^*A \geq A_1 I$, we have $$\| b\|^{2} = \langle A x, A x\rangle = \langle A^* A x, x \rangle \geq A_{1} \|x\|^{2}.$$   Then combining this with $(\ref{e14})$ gives the result. 
\end{proof}

%Now we characterize the bounded operators $A$ for which the Kaczmarz algorithm always converges to a solution.  

\begin{corollary}\label{cor7}
Let $A : \ell^2(\mathbb{N}) \rightarrow \ell^2(\mathbb{N}) $ be a bounded (matrix) operator.  Then, for the initial guess (\ref{e6}), 
the Kaczmarz algorithm always converges to a solution if and only if $A$ is surjective with rows that form an orthogonal basis for $\ell^2(\mathbb{N})$. 
\end{corollary}
\begin{proof}

Let $\{a_n\}_{n = 0}^{\infty} \subset \ell^2(\mathbb{N})$, $b = (b_n)_{n = 0}^{\infty} \in \ell^2(\mathbb{N})$, and for all $n \geq 0$, define $e_n := \frac{a_n^*}{\|a_n\|^2}$.  

($\Leftarrow$)

Let $b\in \ell^2(\mathbb{N})$.  Because $A$ is surjective, there exists $x\in \ell^2(\mathbb{N})$ such that $A x = b$.  Because $\{a_n\}$ is an orthogonal basis, we know $\{e_n\}$ is an orthonormal basis.  In particular, it generates the sequence $\{g_n = e_n\}_{n = 0}^{\infty}$ by formula $(\ref{e2})$, which is a 1-tight frame, so Theorem $\ref{th1}$ proves $\{e_n\}$ is effective.  Then for the initial guess $x_0 = b_0 \frac{e_0}{\|a_0\|}$, formula $(\ref{e31})$ shows $$\lim_{n\rightarrow\infty} \|x_n - x\| = 0. $$  

($\Rightarrow$)

We know $A$ is surjective by the hypothesis that the Kaczmarz algorithm always converges to a solution. 
This same hypothesis also implies $\{e_n = \frac{a_n^*}{\|a_n\|}\}_{n = 0}^{\infty}$ is effective, so Proposition $\ref{l1}$ implies $\{e_n\}$ is a frame.  Hence, the Bessel map $L_e$ is injective, so $$0 = \langle e_n, x\rangle = \langle \frac{a_n^*}{\|a_n\|}, x\rangle = \frac{1}{\|a_n\|} \langle a_n^*, x\rangle$$ for all $n \in \mathbb{N}$ implies $x = 0$.   Observe $A$ is bounded, so $\{a_n\}$ is a Bessel sequence, and we now also have that $A$ is injective.  Then the inverse mapping theorem implies $A$ is an isomorphism, and therefore, $A^*$ is an isomorphism as well.

As $A^*$ is an injective, bounded operator into $\ell^2(\mathbb{N})$, it is a Bessel map for the frame operator $A A^*$.  And because $A^*$ is surjective, it follows that $\{a_n^*\}$ is a Riesz basis.  Therefore, there is an isomorphism $K : \ell^2(\mathbb{N}) \rightarrow \ell^2(\mathbb{N})$ and an orthonormal basis $\{k_n\}_{n \in \mathbb{N}}\subset \ell^2(\mathbb{N})$ such that $K k_n = a_n$ for all $n$.  This implies $\{\|a_n\|\}_{n = 0}^{\infty}$ is bounded below.  

Let $(c_n) \in \ell^2(\mathbb{N})$ and suppose $$0 = \sum_{n = 0}^{\infty} c_n e_n = \sum_{n = 0}^{\infty} \frac{c_n}{\|a_n\|} a_n^*.$$  Because $\{\|a_n\|\}$ is bounded below, we have $(\frac{c_n}{\|a_n\|}) \in \ell^2(\mathbb{N}),$ and as $\{a_n^*\}$ is a basis, we conclude $\frac{c_n}{\|a_n\|} = 0$ for all $n$, which implies $c_n = 0$ for all $n$.  Hence, $\{e_n\}$ is a basis, and because it is also a frame we conclude $\{e_n\}$ is a Riesz basis.  So Theorem $\ref{th10}$ shows $\{e_n\}$ is an orthonormal basis. 
\end{proof}

\subsection*{Acknowledgements}
This research was partly supported by NSF (grant no. CBET 0854233) and ONR (grant no. N000140910144).

%\appendix

%\section*{References}

\end{document}